\documentclass[11pt,a4paper,twoside]{article}
\usepackage[centertags]{amsmath}%ª
\usepackage{color}
\usepackage{amsthm,enumerate}
\usepackage[psamsfonts]{amsfonts,amssymb}
\usepackage{booktabs} 
\usepackage[utf8]{inputenc}

\DeclareMathOperator{\Tr}{Tr}

\newcommand{\SU}{\mathrm{SU}}

\newcommand{\EE}[1]{\mathbb E\left[#1\right]} 
\newcommand{\mi}{\mathrm{i}} %unitˆ immaginaria
\newcommand{\de}{\mathrm{d}} %differenziale
\newcommand{\kb}[2]{\left| #1 \right\rangle \! \left\langle #2 \right|} %ket bra
\newcommand{\ket}[1]{\left| #1 \right\rangle} %ket 
\newcommand{\bra}[1]{\langle #1 |} % bra
\newcommand{\scalar}[2]{\langle #1 | #2\rangle} %prodotto scalare
\newcommand{\braket}[1]{\langle #1 \rangle} % bra ket

\newtheorem{definition}{Definition}
\newtheorem{proposition}[definition]{Proposition}
\newtheorem{lemma}[definition]{Lemma}
\newtheorem{corollary}[definition]{Corollary}
\newtheorem{rem}[definition]{Remark}
\newenvironment{remark}{\begin{rem}  \rm }{\end{rem}}
\newtheorem{rems}[definition]{Remarks}

\newtheorem{example}{Example}

\newtheorem{theorem}{Theorem}

\begin{document}

\title{Quantum game theory for 2$\times$2 games: a mathematical framework}

\author{Gloria Ferraris$^{(1)}$ and Veronica Umanit\`a$^{(2)}$}
\maketitle

\begin{abstract}
We develop a rigorous mathematical framework for quantum game theory applied to static $2\times2$ games, extending classical concepts to the quantum setting where players may employ arbitrary unitary operations (pure strategies) or probability measures over the continuous group
$\mathrm{SU}(2)$ (mixed strategies). We introduce the Eisert-Wilkens-Lewenstein protocol as the standard implementation of such games and prove the existence of Nash equilibria for continuous quantum mixed strategies via a fixed-point argument, extending the classical Nash theorem to the quantum case. 
\end{abstract}


\vskip 1truecm

\section{Introduction}
Game theory is a fundamental mathematical tool for studying strategic interactions among rational agents, with applications ranging from economics to evolutionary biology, and from political science to computer science. First formalized by von Neumann and Morgenstern \cite{von_neumann_morgenstern} and subsequently expanded by the concept of the Nash equilibrium \cite{nash_equilibrium}, classical game theory provides a solid framework for analyzing competitive and cooperative situations in which an agent's decisions influence and are influenced by those of others.

However, since the fundamental laws governing our world are quantum in nature, one naturally wonders what would happen if players had access to quantum resources. Classical theory, in fact,
has certain inherent limitations that have led, since the late 1990s, to the development of a theory of
quantum games, in which players' pure strategies are represented by unitary transformations on quantum
states and the correlations between them can be enhanced by two fundamental resources: superposition
and entanglement. Building on the seminal works \cite{meyer_penny}, \cite{eisert_2009}, \cite{marinatto_weber}, several approaches to extending arbitrary finite $n$-players games to the quantum domain have been proposed (see e.g. \cite{flitney_abbott}, \cite{Jiangfeng_Du_2003}, \cite{bleiler}, \cite{lee}, \cite{chappell}, \cite{das2023}).

The most general and established framework for the quantum implementation of static,  two-players two-strategies games is the Eisert-Wilkens-Lewenstein (EWL) protocol \cite{eisert_quantum}, which introduces quantum correlations between the initial states of the players through a unitary entangling operator. This scheme was subsequently extended to two-players games with $n$ strategies (see \cite{Hayden}), and to multiplayer games with two strategies (see \cite{bolonek1}). 

One of the central problems in game theory is to determine Nash equilibria, that is the optimal outcomes of a game. 
In the context of EWL games, it has been shown that, depending on the strategic space available to each player and the degree of entanglement introduced, new Nash equilibria can emerge that are more advantageous for both players (see e.g. \cite{eisert_2009}, \cite{eisert_quantum}, \cite{marinatto_weber}, \cite{flitney_hollenberg}). In general, as in the classical case, Kakutani's Theorem guarantees the existence of at least a Nash equilibrium for any static quantum game when both players have access to discrete mixed quantum strategies \cite{lee}. However, it is more natural, and mathematically richer, to allow players access to a full set of quantum moves, i.e. to the set of probability measures over the unitary group of pure strategies. For maximally entangled EWL games, Landsburg \cite{landsburg} proved that any mixed strategy is equivalent to one supported on at most four pure strategies, thereby reducing the problem to the discrete case i.e. to a discrete mixed strategies (see also \cite{bolonek1} and \cite{bolonek2}). Yet, when the entanglement is not maximal, no general existence result for Nash equilibria in the full continuous strategy space has been established. This gap constitutes the main motivation for the present work.
\subsection{Contributions and Outlines}

The aim of this paper is to provide a unified and rigorous treatment of quantum game theory for $2\times 2$ games, with a particular focus on mathematical formalization and the search for strategic equilibria in a fully continuous setting. Our main contributions are the following:

\begin{itemize}
\item We introduce a definition of continuous quantum mixed strategies as arbitrary probability measures on the entire group $\SU(2)$, encompassing both discrete and genuinely continuous mixtures of unitary operations.
\item We prove the existence of Nash equilibria when both players have access to such continuous mixed strategies (Theorem \ref{esistenza}). This result holds for every $2\times 2$ game with continuous payoff functions and is independent of the specific quantization protocol or the degree of entanglement. In particular, it can be applied to \emph{all} EWL games (see Theorem \ref{esistenza-EWL}).
\item We introduce analytic tools, including a quadratic-form representation of payoffs (see \cite{Jiangfeng_Du_2003}), that facilitate the explicit computation and analysis of equilibria within the EWL framework.
\end{itemize} 

The structure of the paper is as follows. Section $2$ outlines the formalism of classical \(2 \times 2\) game theory and introduces the concept of a Nash equilibrium. Section $3$ shows how this classical framework admits a natural quantum generalization, in which the strategy space is extended from discrete sets to the unitary group $\mathrm{SU(2)}$. We then define quantum mixed strategies as probability measures on $\mathrm{SU(2)}$ and introduce the corresponding notion of quantum equilibrium. Section 4 provides an existence theorem for Nash equilibria in continuous quantum mixed strategies, generalising the classical result. Section $5$ presents the Eisert-Wilkens-Lewenstein protocol \cite{eisert_quantum} and proves the existence of Nash equilibria within this framework (Theorem \ref{esistenza-EWL}).

Finally, in Section $5.4$ we also analyse an asymmetric scenario of the Chicken (Hawk-Dove) game, demonstrating that a quantum player can always achieve a higher payoff than a classical opponent by suitably choosing their quantum strategy.

\section{Classical framework for 2\,x\,2 games}
 Game theory is the branch of mathematics that studies strategic interactions between rational agents, called \emph{players}. The possible actions of players are called \emph{moves}, and a prescription that specifies the particular move to be made in all possible game situations is a \emph{strategy}. A \emph{pure strategy} is a deterministic choice of a single action. In particular, in a static game, i.e., where each player makes a single, simultaneous decision without knowing others' choices, a pure strategy coincides with a single move.

We focus on the simplest non-trivial setting: two-player games where each player has exactly two pure  strategies (\emph{$2\times 2$ game}).  
Let \(S=\{s_1,s_2\}\) be the strategy set of player \(A\) and \(T=\{t_1,t_2\}\) that of player \(B\).  
A pair \((s,t)\in S\times T\) is called a \emph{strategy profile}.

Each player has a \emph{payoff function}
\[
a:S\times T\to\mathbb{R},\qquad b:S\times T\to\mathbb{R},
\]
where \(a(s,t)\) (resp. \(b(s,t)\)) is the payoff of player \(A\) (resp. \(B\)) when \(A\) chooses \(s\) and \(B\) chooses \(t\).  
The game can be represented by a \(2\times2\) bimatrix
\[
\begin{array}{c|cc}
& t_1 & t_2 \\ \hline
s_1 & (a_{11},b_{11}) & (a_{12},b_{12})\\
s_2 & (a_{21},b_{21}) & (a_{22},b_{22})
\end{array}
\]
with \(a_{ij}=a(s_i,t_j),\; b_{ij}=b(s_i,t_j)\).

\begin{definition}[Mixed strategy]
A \emph{mixed strategy} for player \(A\) is a probability vector \((p,1-p)\) with \(p\in[0,1]\), meaning that \(s_1\) is chosen with probability \(p\) and \(s_2\) with probability \(1-p\).  
Analogously, a mixed strategy for player \(B\) is \((q,1-q)\) with \(q\in[0,1]\).
\end{definition}
In other words, the mixed strategy $(p,1-p)$ for $A$ is the probability measure $\mu_{A,p}$ on $S$ given by 
$$\mu_{A,p}=p\delta_{s_1}+(1-p)\delta_{s_2},$$ 
while the mixed strategy $(q,1-q)$ for B corresponds to 
$$\mu_{B,q}=q\delta_{t_1}+(1-q)\delta_{t_2}.$$
We indicate by $M_S$ and $M_T$ the sets of mixed strategies for $A$ and $B$ respectively, and we will use any of the following notations interchangeably
\begin{align*}
    M_S&=\{\mu_{A,p}\,|\, p\in[0,1]\}=\{(p,1-p)\,|\,p\in[0,1]\},\\ M_T&=\{\mu_{B,q}\,|\, q\in[0,1]\}=\{(q,1-q)\,|\,p\in[0,1]\}.
\end{align*}
Clearly, pure strategies can be considered as mixed strategies in which the probability is $0$ or $1$, i.e. $S=\{\delta_{s_1}, \delta_{s_2}\}\subseteq M_S$ and $T=\{\delta_{t_1},\delta_{t_2}\}\subseteq M_T$.

In this way, by independently setting two mixed strategies $\mu_{A,p}\in M_S$ and $\mu_{B,q}\in M_T$, we may think of players $A$ and $B$ as discrete independent random variables (which we denote by $A_p$ and $B_q$) with values in $S$ and $T$ respectively, having $\mu_{A,p}$ and $\mu_{B,q}$ as distribution laws. Consequently, the payoff functions become the random variables $a(A_p,B_q)$ and $b(A_p,B_q)$, and we are interested in determining their expected values
\begin{align*}
\EE{a(A_p,B_q)}&= p\,q\,a_{11} + p(1-q)a_{12} + (1-p)q\,a_{21} + (1-p)(1-q)a_{22}\\
&=\sum_{i,j=1,2}a(s_i,t_j)\mu_{A,p}(\{s_i\})\mu_{B,q}(\{t_j\})=:E_A(p,q)
\end{align*}
and 
\begin{align*}
\EE{b(A_p,B_q)}&=p\,q\,b_{11} + p(1-q)b_{12} + (1-p)q\,b_{21} + (1-p)(1-q)b_{22}\\
&=\sum_{i,j=1,2}b(s_i,t_j)\mu_{A,p}(\{s_i\})\mu_{B,q}(\{t_j\})=:E_B(p,q).
\end{align*}

\begin{definition}[Dominant strategy]
A pure strategy \(s^*\in S\) is \emph{dominant} for player \(A\) if, for every \(t\in T\),
\[
a(s^*,t)\ge a(s,t)\quad\forall s\in S.
\]
If the inequality is strict for at least one \(t\), the strategy is \emph{strictly dominant}.  
The definition for player \(B\) is analogous.
\end{definition}

\begin{definition}[Nash equilibrium]
A strategy profile $(p^*,q^*)\in M_S\times M_T$ is a \emph{Nash equilibrium} (NE) if no player can improve his payoff by unilaterally deviating, i.e.
\[
E_A(p^*,q^*)\ge E_A(p,q^*)\ \ \forall p\in[0,1],\qquad 
E_B(p^*,q^*)\ge E_B(p^*,q)\ \ \forall q\in[0,1].
\]
\end{definition}

A convenient characterisation of mixed-strategy equilibria uses the notion of \emph{indifference}:  
in equilibrium each player's mixing must make the opponent indifferent between his two pure strategies.  
Formally, \((p^*,q^*)\) is a mixed-strategy NE only if
\[
E_A(0,q^*)=E_A(1,q^*),\quad 
E_B(p^*,0)=E_B(p^*,1).
\]  
Explicitly,
\[
\left\{\begin{array}{l}q^*a_{11}+(1-q^*)a_{12}=q^*a_{21}+(1-q^*)a_{22},\\
p^*b_{11}+(1-p^*)b_{21}=p^*b_{12}+(1-p^*)b_{22}.
\end{array}\right.
\]

\begin{definition}[Pareto optimality]
A profile \((s^*,t^*)\) is \emph{Pareto optimal} if there is no other profile \((s,t)\) such that
\[
a(s,t)\ge a(s^*,t^*),\quad b(s,t)\ge b(s^*,t^*),
\]
with at least one inequality strict.  
Equivalently, one player cannot be made better off without making the other worse off.
\end{definition}

\section{Quantum model for 2\,x\,2 games}

The classical model of a game with two (pure) strategies has only two possible configurations (or states), and players can only leave them unchanged or flip them. Identifying these configurations with the canonical orthonormal basis $\ket{0}$, $\ket{1}$ of $\mathbb{C}^2$, classical pure strategies correspond to permutations of these states, i.e. to the $2\times 2$ matrices 
$$I=\left(\begin{array}{cc}
1&0\\
0&1
\end{array}\right),\qquad F=\left(\begin{array}{cc}
0&1\\
1&0
\end{array}\right),
$$
while the mixed strategies $(p,1-p)$ are described by convex combinations $pI+(1-p)F$ for $p\in[0,1]$.

In a quantum extension of the game, the system is no longer confined to a definite classical state but may exist in a superposition of states. A general quantum state of a single two-level system (a qubit) can be written as a vector $\ket{v}=a\ket{0}+b\ket{1}$ with $a,b\in\mathbb{C}$ such that $|a|^2+|b|^2=1$, where $|a|^2$ and $|b|^2$ represent the probabilities that a measurement will yield the outcomes corresponding to $\ket{0}$ and $\ket{1}$, respectively. In an equivalent way, a qubit $\ket{v}$ can be described by the pure-state density operator $$\rho_v:=\kb{v}{v}=|a|^2\kb{0}{0}+a\overline{b}\kb{0}{1}+b\overline{a}\kb{1}{0}+|b|^2\kb{1}{1}.$$
If we perform a projective measurement in the computational basis $\{\ket{0},\ket{1}\}$, described by the PVM $\{P_0, P_1\}$, $P_i=\kb{i}{i}$ for $i=0,1$, the probability of outcome $i$ is 
$$p_v(i)=\Tr(\rho_v P_i)=|\scalar{i}{v}|^2.$$ We also refer to this measurement as the measurement on the \emph{computational basis} $\{\ket{0}, \ket{1}\}$.

Assuming that players can manipulate their qubits without introducing decoherence, quantum pure strategies correspond to unitary operations on $\mathbb{C}^2$. However, since global phases have no observable effect on quantum states, it suffices to consider unitary matrices with determinant one, i.e. elements of the group $\SU(2)$.
Therefore, the action of a single player on the initial state $\ket{v}$ is given by $U\ket{v}$ for a suitable $U\in\SU(2)$.

Thus, the space of \emph{quantum pure strategies} is continuous, can be parameterised as
\[\left\{
U(\theta,\alpha,\beta)=
\begin{pmatrix}
e^{\mi\alpha}\cos\frac\theta2 & \mi e^{\mi\beta}\sin\frac\theta2\\
\mi e^{-\mi\beta}\sin\frac\theta2 & e^{-\mi\alpha}\cos\frac\theta2
\end{pmatrix}\,|\, \theta\in[0,\pi],\; \alpha,\beta\in[0,2\pi)\right\},
\]
and $U(\theta,\alpha,\beta)^\dagger=U(\theta,\alpha,\beta)^{-1}=U(\theta,2\pi-\alpha,\pi+\beta)$.\\

Using the expansion in the Pauli basis $I,\mi\sigma_x,\mi\sigma_y,\mi\sigma_z$, 
$$\sigma_x=\left(\begin{array}{cc}
0&1\\
1&0
\end{array}\right),\quad \sigma_y=\left(\begin{array}{cc}
0&-\mi\\
\mi&0
\end{array}\right),\quad \sigma_z=\left(\begin{array}{cc}
1&0\\
0&-1
\end{array}\right),$$
every $U\in\SU(2)$ can be written uniquely as
\[
U = w I + \mi x \sigma_x + \mi y \sigma_y + \mi z \sigma_z
\]
with $w,x,y,z\in\mathbb{R}$ such that $w^2+x^2+y^2+z^2=1$.
Comparing with the parameterisation above gives the explicit correspondence
\[
w = \cos\alpha\cos\frac\theta2,\quad
x = \cos\beta\sin\frac\theta2,\quad
y = -\sin\beta\sin\frac\theta2,\quad
z = \sin\alpha\cos\frac\theta2.
\]
Thus, each $U\in\SU(2)$ is uniquely associated with a unit vector
\[
u = (w,x,y,z) \in \mathbb{R}^4,
\qquad \|u\|=1,
\]
which provides a convenient real-parameter description of single-qubit strategies in quantum games.

\medskip

\begin{remark}
\rm    
Pure classical strategies correspond (up to a phase) to special cases of the quantum parameterisation:
\[
I = U(0,0,0), \qquad F = \mi\sigma_x = U(\pi,0,0).
\]
\end{remark} 

Mixed classical strategies can also be seen as special cases of pure quantum strategies.
Indeed, if we start from the initial state \(\ket{v}=a\ket{0}+b\ket{1}\) with \(a,b\in\mathbb{R}\), the mixed strategy $(p,1-p)$ produces $\ket{v}$ with probability $p$ and $\ket{v_F}:=F\ket{v}=a\ket{1}+b\ket{0}$ with probability $1-p$. Therefore, the corresponding state is the mixture 
$$\rho_{\text{c}}:=p\ket{v}\!\bra{v}+(1-p)\ket{v_F}\!\bra{v_F},$$
and the probability of outcome \(\ket{0}\) is then
\(p_{\text{c}}(0)=p|a|^{2}+(1-p)|b|^{2}\).  

On the other hand, applying the quantum operator \(U(\theta,0,0)\) to $\ket{v}$ yields the pure state  
$$\ket{w}=U(\theta,0,0)\ket{v}=\bigl(a\cos\frac{\theta}{2}+\mi b\sin\frac{\theta}{2}\bigr)\ket{0}
       +\bigl(b\cos\frac{\theta}{2}+\mi a\sin\frac{\theta}{2}\bigr)\ket{1},
$$
so that  
\(p_{\text{q}}(0)=|\braket {0|w}|^{2}
 =|a|^{2}\cos^{2}\frac{\theta}{2}+|b|^{2}\sin^{2}\frac{\theta}{2}
  +2\Re\!\bigl(\mi a\overline{b}\cos\frac{\theta}{2}\sin\frac{\theta}{2}\bigr)\).  
Since \(a,b\in\mathbb{R}\), the interference term vanishes, and with the identification  
$$p=\cos^{2}(\theta/2),\ 1-p=\sin^{2}(\theta/2)$$ we obtain \(p_{\text{q}}(0)=p_{\text{c}}(0)\).  
The same holds for outcome \(\ket{1}\); hence, we have proven the following result.

\begin{proposition}
Whenever the initial state has real coefficients, the quantum pure strategy 
\[
U(\theta,0,0)=\cos\frac\theta2\,I + \mi\sin\frac\theta2\,\sigma_x,\qquad\theta\in(0,\pi)
\]
reproduces exactly the statistics of the classical mixed-strategy $(p,1-p)$ with $p:=\cos^2\frac\theta2$.
\end{proposition}

This shows that the quantum framework naturally contains the classical one as a proper subset.

\medskip

In full analogy with classical game theory, one can define \emph{quantum mixed strategies} as probabilistic choices among unitary operators. 
Because the set of pure strategies \(\SU(2)\) is a compact continuous group, a mixed strategy is formally a probability measure \(\mu\) on the Borel \(\sigma\)-algebra of \(\SU(2)\).  \\
We denote by $\mathcal{M}$ the space of all such strategies,
\begin{equation}\label{def-M}
\mathcal{M} = \bigl\{\mu:\mathcal{B}(\SU(2))\to[0,1] \mid \mu \text{ is a regular Borel probability measure}\bigr\}.
\end{equation}

A practically important subclass consists of \emph{discrete mixed strategies}, where a player selects a finite set \(\{U_1,\dots,U_N\}\subset\SU(2)\) with probabilities \(p_1,\dots,p_N\) (\(p_k\ge0,\ \sum_k p_k=1\)).
Such a strategy corresponds to the Dirac-measure combination $$\mu = \sum_k p_k\delta_{U_k},$$
and induces the quantum channel (i.e. the completely positive and trace-preserving map) on $M_2(\mathbb{C})$
$$\Phi_\mu(\rho)=\sum_k p_k\,U_k\rho U_k^*.$$

Thus, the formalism of density matrices and quantum channels provides a unified description covering both classical probabilistic mixtures and genuinely quantum superpositions.
\smallskip\\

In a static quantum $2\times 2$ game, players $A$ and $B$ each control a qubit and they act simultaneously on it, so that the joint state lies in $\mathbb{C}^2\otimes\mathbb{C}^2$. Adapting from \cite{eisert_2009}, we can then introduce the following definition.

\begin{definition}
A \emph{static $2\times 2$ quantum game} $\mathbf{G}$ is a tuple
\[
\mathbf{G}= (\mathbb{C}^2\otimes\mathbb{C}^2, \rho, S_A, S_B, P_A, P_B),
\]
where:
\begin{itemize}
\item[-] $\mathbb{C}^2\otimes\mathbb{C}^2$ is the Hilbert space describing the joint physical system of the two players,
\item[-] $\rho \in \mathcal{S}(\mathbb{C}^2\otimes\mathbb{C}^2)$ is the initial state of the system,
\item[-] $S_A, S_B \subseteq \mathrm{SU}(2)$ are the sets of pure strategies available to players $A$ and $B$, respectively,
\item[-] the \emph{payoff operators} $P_A, P_B$ are self-adjoint operators on $\mathbb{C}^2\otimes\mathbb{C}^2$ associated with the measurement performed. They specify the utility for each player, and depend on the classical payoffs $a$ and $b$, respectively. More precisely,
\begin{equation}\label{def-PX} P_X=\sum_{i,j=0,1}x_{ij}|\psi_{ij}\rangle\langle\psi_{ij}|
\end{equation}
with \(x_{ij}\) the classical payoffs of player $X=A,B$, and $\{|\psi_{ij}\rangle\langle\psi_{ij}|\,|\,i,j=0,1\}$ a PVM.
\end{itemize}

If player \(A\) uses the pure strategy \(U_A\in S_A\) and player \(B\) uses \(U_B\in S_B\), the final state of the system is
$$\rho_f(U_A,U_B):=(U_A\otimes U_B)\rho(U_A\otimes U_B)^\dagger=(U_A\otimes U_B)\rho(U_A^\dagger\otimes U_B^\dagger).$$

\end{definition}
\smallskip

The choice of the initial state $\rho$ and of the PVM $\{|\psi_{ij}\rangle\langle\psi_{ij}|\,|\,i,j=0,1\}$ depends on the protocol used to quantize the game, while $S_A$ and $S_B$ are determined by the quantum operations to which the players have access. For example, in \cite{eisert_2009} and \cite{eisert_quantum} the Prisoner's Dilemma was analysed considering different strategic spaces, including $\{U(\theta,0,0)\,|\,\theta\in[0,\pi]\}$ and $\{U(\theta,\alpha,0)\,|\,\theta\in[0,\pi], \alpha\in[0,2\pi)\}$. However, many authors (see e.g. \cite{benjamin} and \cite{flitney_hollenberg}) have pointed out that restricting the strategy space to proper subsets of $\SU(2)$ represents a serious limitation and is not physically justifiable in the absence of a physical motivation. Therefore, in the following, by a static $2\times2$ game we will always mean a game with pure strategy spaces $S_A=S_B=\SU(2)$. The results obtained are, however, immediately generalisable to the case where $S_A$ and $S_B$ are locally compact subsets of $\SU(2)$ (see Remark 4.3).

\subsection{Quantum expected payoff}
Let $\mathbf{G}= (\mathbb{C}^2\otimes\mathbb{C}^2, \rho, \SU(2), \SU(2), P_A, P_B)$ be a
static $2\times 2$ quantum game, with $P_A$ and $P_B$ given by \eqref{def-PX}. We denote by $\mathcal{M}_A$ and $\mathcal{M}_B$ the sets of associated mixed quantum strategies, with $\mathcal{M}_A=\mathcal{M}_B=\mathcal{M}$ defined in \eqref{def-M}.

From now on, by a static $2\times 2$ quantum game we will always mean the game $\mathbf{G}$ as defined above.
\begin{definition}
If player \(A\) uses the pure strategy \(U_A\in\SU(2)\) and player \(B\) uses \(U_B\in\SU(2)\), the \emph{expected payoff} for player \(X\) (\(X=A,B\)) is
\begin{align*}
\langle\$_X(U_A,U_B)\rangle &= \Tr\bigl(\rho_f(U_A,U_B)\,P_X\bigr)\\
%&=\sum_{i,j\in\{0,1\}} x_{ij}\Tr\bigl(\rho|(U_A\otimes U_B)^\dagger\psi_{ij}\rangle\langle(U_A\otimes U_B)^\dagger\psi_{ij}|\bigr)\\
&=\sum_{i,j\in\{0,1\}} x_{ij}\langle\psi_{ij}|(U_A\otimes U_B)\rho(U_A\otimes U_B)^\dagger|\psi_{ij}\rangle\\
&=\sum_{i,j\in\{0,1\}} x_{ij}\|\rho^\frac{1}{2}(U_A\otimes U_B)^\dagger|\psi_{ij}\rangle\|^2.
\end{align*}
\end{definition}

Hence, the quantum payoff is the expectation value of the observable \(P_X\) in the final state, i.e. a weighted average of the classical payoffs with weights given by the probabilities of observing the outcomes \((i,j)\) when measuring in the basis $\ket{\psi_{ij}}$, $i,j=0,1$.

If the final state is pure, \(\rho_f = |v_f\rangle\langle v_f|\), the probabilities become \(|\langle \psi_{ij}|v_f\rangle|^2\) and
\begin{equation}\label{eq:payoff-puro}
\langle\$_X(U_A,U_B)\rangle = \sum_{i,j} x_{ij}\,|\langle \psi_{ij}|v_f\rangle|^2.
\end{equation}

The definition extends naturally to mixed strategies.  
\begin{definition}
If \(\mu_A,\mu_B\in\mathcal{M}\) are quantum mixed strategies, the expected payoff for player \(X=A,B\) is
\[
\langle\$_X(\mu_A,\mu_B)\rangle = 
\iint_{\SU(2)\times\SU(2)} \Tr\!\bigl(\rho_f(U_A,U_B)P_X\bigr)\,
\de\mu_A(U_A)\,\de\mu_B(U_B).
\]
\end{definition}
This formula covers all cases: pure strategies (Dirac measures), discrete mixed strategies (finite convex combinations), and general mixed strategies. For discrete mixed strategies we will use the notations
$$\mu_X=\sum_{j=1}^Nq_j\delta_{U_j}\quad\mbox{or}\quad\mu_X=\sum_{j=1}^Nq_j{U_j}$$ interchangeably.
\begin{remark}
    \rm
    Since $P_X$ is a self-adjoint operator having $x_{ij}$ as eigenvalues, its norm is given by $\|P_X\|=\max_{i,j=0,1}|x_{ij}|$. As a consequence we have
    $$|\langle\$_X(U_A,U_B)\rangle|\leq\|\rho_f(U_A,U_B)\|_1\,\|P_X\|=\max_{i,j=0,1}|x_{ij}|,$$
    and then the expected payoff is uniformly bounded on $\mathcal{M}$,
    $$|\langle\$_X(\mu_A,\mu_B)\rangle|\leq \max_{i,j=0,1}|x_{ij}|,
    $$
    being $\mu_A$ and $\mu_B$ probability measures.\end{remark}

Finally, we introduce the \emph{best-response correspondence}, $$\mathcal{B}:\mathcal{M}_A\times\mathcal{M}_B\to 2^{\mathcal{M}_A\times\mathcal{M}_B},\ \qquad
\mathcal{B}(\mu_A,\mu_B) = \mathcal{B}_A(\mu_B) \times \mathcal{B}_B(\mu_A),
$$
which maps an opponent's strategy to the set of strategies that maximise one's own payoff.
For player \(A\) the correspondence is
\[
\mathcal{B}_A(\mu_B)=\Bigl\{
\mu_A\in\mathcal{M}\; \Big|\; 
\langle\$_A(\mu_A,\mu_B)\rangle = \sup_{\nu_A\in\mathcal{M}}\langle\$_A(\nu_A,\mu_B)\rangle
\Bigr\},
\]
and symmetrically for player \(B\).

\section{Quantum Nash equilibrium}
The definition of Nash equilibrium extends naturally to the quantum case. We will denote by $\mathcal{M}_X$ the space of quantum mixed strategies of player $X$ for $X=A,B$.
\begin{definition}
A pair $(\mu_A^\star, \mu_B^\star) \in \mathcal{M}_A \times \mathcal{M}_B$ is a \emph{quantum Nash equilibrium} if
\[
\begin{cases}
\langle\$_A(\mu_A^\star, \mu_B^\star)\rangle \geq \langle\$_A(\mu_A, \mu_B^\star)\rangle \quad \forall \mu_A \in \mathcal{M}_A, \\
\langle\$_B(\mu_A^\star, \mu_B^\star)\rangle \geq \langle\$_B(\mu_A^\star, \mu_B)\rangle \quad \forall \mu_B \in \mathcal{M}_B.
\end{cases}
\]
\end{definition}
Similarly to the classical case, it is clear that Nash equilibria are precisely the fixed points of $\mathcal{B}$, i.e. the pairs $(\mu_A^\star, \mu_B^\star)$ such that $(\mu_A^\star, \mu_B^\star) \in \mathcal{B}(\mu_A^\star, \mu_B^\star)$.

In order to prove the existence of such points, we will make use of the following result that extends Kakutani's fixed point Theorem to multifunctions on locally convex spaces (see \cite{arandelovic2000}).
\begin{theorem}[Kakutani-Glicksberg-Ky Fan]\label{teo_kgk}
Let $X$ be a $\text{T}_2$ locally convex topological vector space and $K \subseteq X$ a non empty, compact and convex subset. If $F: K \to 2^K$ is a multifunction such that:
\begin{enumerate}
    \item the graph $\mathcal{G}(F):=\{(x,y)\,|\, y\in F(x)\}$ is closed in $K\times K$,
    \item $F(x)$ is non empty and convex for all $x\in K$,
\end{enumerate}
then $\operatorname{Fix}F\neq \emptyset$, where $\operatorname{Fix}F$ denotes the set of fixed point of $F$.
\end{theorem}

Our goal is to apply Theorem \ref{teo_kgk} to 
\begin{align*}
X &= M(\mathrm{SU(2)}) := \{ \text{Borel regolar measures on}\ \mathrm{SU(2)} \}\\
K &= \mathcal{M} = \{\mu \in M(\mathrm{SU(2)}) \mid \mu \geq 0,\ \mu(\mathrm{SU(2)}) = 1\}\\
F &= \mathcal{B} : \mathcal{M} \times \mathcal{M} \to 2^{\mathcal{M} \times \mathcal{M}}.
\end{align*}
Since $\mathrm{SU(2)}$ is a compact normed space, the Banach space $M(\mathrm{SU(2)})$ is isometrically isomorphic to the dual of continuous functions $\mathcal{C}(\mathrm{SU(2)})$ on $\mathrm{SU(2)}$, and so it is a $\text{T}_2$ locally convex topological vector space with respect to the weak$^*$ topology. Moreover, $\mathcal{M}$ is non empty (it contains measures $\delta_U$ for $U\in\mathrm{SU(2)}$), convex and weakly$^*$ compact by Banach-Alaoglu-Bourbaki Theorem. Therefore, in order to apply Theorem \ref{teo_kgk}, we only need $\mathcal{B}$ to satisfy conditions $1$ and $2$. 

\begin{remark}\rm
If the space of pure strategies $S_X$ is not all $\SU(2)$ but is locally compact, we can still define a mixed quantum strategy as a regular Borel probability measure on it, so that
$$\widetilde{\mathcal{M}}:=\{\mbox{quantum mixed strategies}\}\subseteq M(S_X).
$$
In this case, $M(S_X)$ is isometrically isomorphic to the topological dual space of $\mathcal{C}_0(S_X)$, the set of continuous functions on $S_X$ which vanish at infinity, and $\widetilde{\mathcal{M}}$ is still non empty, convex and weakly$^*$ compact.
\end{remark}
We introduce the following preliminary results.

\begin{lemma}\label{lem:cont}
The payoff function
\[
(U_A,U_B)\longmapsto\langle\$_{X}(U_A,U_B)\rangle
   =\sum_{i,j\in\{0,1\}} x_{ij}\|\rho^\frac{1}{2}(U_A\otimes U_B)^\dagger|\psi_{ij}\rangle\|^2
\]
is continuous on $\SU(2)\times\SU(2)$ for all $X=A,B$.
\end{lemma}

\begin{proof}
Convergences
$U_{A,n}\to U_A$ and $U_{B,n}\to U_B$ in operator norm imply
$\rho^\frac{1}{2}(U_{A,n}\otimes U_{B,n})^\dagger\to \rho^\frac{1}{2}(U_{A}\otimes U_{B})^\dagger$ in the strong operator topology, so that
$$\|\rho^\frac{1}{2}(U_{A,n}\otimes U_{B,n})^\dagger\ket{v}\|\to_n\|\rho^\frac{1}{2}(U_A\otimes U_B)^\dagger|v\rangle\|$$ for all $v\in\mathbb{C}^4$.
Choosing $v=\psi_{ij}$ gives the desired continuity.
\end{proof}

\begin{lemma}\label{contw}
The map
\begin{align*}
    \langle\$_X\rangle : \mathcal{M} \times \mathcal{M} &\to \mathbb{R}\\
    (\mu_A,\mu_B)&\mapsto\iint_{\mathrm{SU(2)}\times \mathrm{SU(2)}} \langle\$_X(U_A, U_B)\rangle \de\mu_A(U_A)\de\mu_B(U_B)
\end{align*}
is w$^*$-continuous for $X=A,B$.
\end{lemma}
\begin{proof} We prove the theorem for $X=A$. The other case is clearly analogous.\\
Let $(\mu_{A_i})_{i \in I}$, $(\mu_{B_i})_{i\in I}\subseteq\mathcal{M}$ such that
$\mu_{A_i} \xrightarrow{w^*} \mu_A$ and $\mu_{B_i} \xrightarrow{w^*} \mu_B$, and prove that $\langle\$_A(\mu_{A_i}, \mu_{B_i})\rangle \rightarrow_i \langle\$_A(\mu_A, \mu_B)\rangle$.

Since the function $U_B \mapsto \int \langle\$_A(U_A, U_B)\rangle \, \de\nu_A(U_A)$ is $\nu_B$-integrabile for all $\nu_A,\nu_B\in\mathcal{M}$ by virtue of inequality
\begin{equation}\label{payoff-lim}
|\langle\$_A(U_A, U_B)\rangle| \leq\max_{i,j\in\{0,1\}} |a_{ij}| \qquad \forall\  U_A, U_B \in \mathrm{SU(2)},
\end{equation}
Fubini's Theorem gives
\begin{align*}
| \langle\$_A(\mu_{A_i}, \mu_{B_i})\rangle - \langle \$_A(\mu_A, \mu_B)\rangle | &\leq
  \left| \int \left( \int \langle\$_A(U_A, U_B)\rangle  \de\mu_{A_i} - \int \langle\$_A(U_A, U_B)\rangle  \de\mu_A \right) \de\mu_{B_i} \right| \\
 &+ \left| \iint \langle\$_A(U_A, U_B)\rangle \de\mu_A \de\mu_{B_i} - \iint \langle\$_A(U_A, U_B)\rangle  \de\mu_A \de\mu_B  \right|\\
 &=\left| \int (g_i-g) \de\mu_{B_i}\right|+\left| \int g\, \de\mu_{B_i} - \int g\, \de\mu_B \right|,
\end{align*}
where we have set
\[ g_i(U_B) := \int \langle\$_A(\cdot, U_B)\rangle \de\mu_{A_i},\qquad g(U_B)=\int \langle\$_A(\cdot, U_B)\rangle \de\mu_A \] 
for all $U_B\in \mathrm{SU(2)}$ and $i\in I$.\\
Now, note that \( g_i,g : \mathrm{SU(2)} \to \mathbb{R} \) are continuous function by Lebesgue's Theorem, since \( \langle\$_A(U_A, \cdot) \rangle\) is continuous and inequality \eqref{payoff-lim} holds. Hence, 
\[\sup_{U_B\in \mathrm{SU(2)}} |g_i(U_B)-g(U_B)| = \max_{U_B \in \mathrm{SU(2)}} |g_i(U_B)-g(U_B)| = |g_i(\overline{U}_B)-g(\overline{U}_B)|\]
for a suitable \( \overline{U}_B \in \mathrm{SU(2)} \), being \( \mathrm{SU(2)} \) a compact set. But $w^*-\lim_i\mu_{A_i}=\mu_A$ and $\langle\$_A(\cdot, \overline{U}_B)\rangle\in\mathcal{C}(\mathrm{SU(2)})$, so that $g_i(\overline{U}_B)\rightarrow_ig(\overline{U}_B)$. This means that, given $\varepsilon>0$, there exists $i_\varepsilon\in I$ such that 
$$\sup_{U_B\in \mathrm{SU(2)}} |g_i(U_B)-g(U_B)| < \frac{\varepsilon}{2}\qquad\forall\,  i \ge i_\varepsilon,$$
and so
$$\left| \int (g_i-g) \de\mu_{B_i}\right|\leq \int |g_i-g|\de\mu_{B_i}< \int \frac{\varepsilon}{2} d\mu_{B_i} = \frac{\varepsilon}{2}\qquad\forall\,  i \ge i_\varepsilon.$$

On the other hand, since $g$ is continuous and $\mu_{B_i}\to\mu_B$ in the weak$^*$-topology on $\mathcal{M}$, there exists $j_\varepsilon \in I$ such that
$$\left| \int g\, d\mu_{B_i} - \int g\, \de\mu_B \right|< \frac{\varepsilon}{2} \qquad \forall\, i \geq j_\varepsilon.$$
Finally, taking $k_\varepsilon \in I$ with $k_\varepsilon \geq i_\varepsilon$ and $k_\varepsilon \geq j_\varepsilon$, we get  
$$| \langle\$_A(\mu_{A_i}, \mu_{B_i})\rangle - \langle \$_A(\mu_A, \mu_B)\rangle |<\varepsilon\qquad\forall\, i \geq k_\varepsilon,$$
which is the thesis. In the same way we can prove the weak$^*$ continuity of $\langle\$_B\rangle$ on $\mathcal{M}$.
\end{proof}

We are now in position to show the existence of at least one Nash equilibrium for quantum mixed strategies.

\begin{theorem}\label{esistenza}
Let $\mathcal{B}$ be the best-response multifunction of a static $2\times 2$ quantum game. Then
$\operatorname{Fix} \mathcal{B} \neq \emptyset$.
\end{theorem}
\begin{proof}
As said before, in order to apply Theorem \ref{teo_kgk} we have to prove that $\mathcal{B}$ has a closed graph and $\mathcal{B}(\mu_A,\mu_B)$ is a non empty convex set for all $\mu_A,\mu_B\in\mathcal{M}$. First of all we note that, since $\mathcal{B}(\mu_A,\mu_B) = \mathcal{B}_A(\mu_B) \times \mathcal{B}_B(\mu_A)$, it is enough to show these properties for each of the multifunctions $\mathcal{B}_X:\mathcal{M}\to 2^\mathcal{M}$, with $X=A,B$. Assume for example $X=A$, the other case is symmetric.

Let $\mu_B \in \mathcal{M}$. Thanks to Lemma \ref{payoff-lim} the map \( \langle\$_A (\cdot, \mu_B)\rangle\) is w$^*$-continuous on the $w^*$-compact set $\mathcal{M}$, and then there exists a measure $\overline{\nu}_A \in \mathcal{M}$ such that 
$$\langle\$_A (\overline{\nu}_A, \mu_B)\rangle = \max_\mathcal{M} \langle\$_A (\cdot, \mu_B)\rangle=:c.$$
Therefore, the set
\[
\mathcal{B}_A (\mu_B) = \{\mu_A \in \mathcal{M} \mid \langle\$_A (\mu_A, \mu_B)\rangle = \sup_{\mathcal{M}} \langle\$_A (\cdot, \mu_B)\rangle\},
\]
can be written as
 \[ \mathcal{B}_A (\mu_B) = \{\mu_A \in \mathcal{M} \mid \langle\$_A (\mu_A, \mu_B)\rangle = c\}, \] 
 and $\nu_A$ belongs to it.
 The convexity of $\mathcal{B}_A (\mu_B)$ immediately follows.

We prove now that $\mathcal{G}(\mathcal{B}_A)$ is weakly$^*$ closed in $\mathcal{M}$.\\
Given \((\mu_{B_i})_i, (\mu_{A_i})_i \subseteq \mathcal{M}\) such that $w^*-\lim_i\mu_{B_i}=\mu_B$,\ $w^*-\lim_i\mu_{A_i}=\mu_A$ and \(\mu_{A_i} \in \mathcal{B}_A (\mu_{B_i})\) for all $i\in I$, we have to show that \(\mu_A\) belongs to $\mathcal{B}_A (\mu_B)$, i.e. \[\langle \$_A (\mu_A, \mu_B)\rangle =\max_{\mathcal{M}} \langle\$_A (\cdot, \mu_B)\rangle.\]
By definition of $\mathcal{B}_A (\mu_{B_i})$ we have
\[ \langle\$_A (\mu_{A_i}, \mu_{B_i})\rangle = \max_\mathcal{M} \langle\$_A (\cdot, \mu_{B_i})\rangle \geq \langle\$_A (\nu_A, \mu_{B_i})\rangle \quad \forall\, i\in I,\  \forall\, \nu_A \in \mathcal{M},\]
and $$\langle\$_A(\mu_{A_i}, \mu_{B,i})\rangle \xrightarrow[i]{} \langle\$_A(\mu_A, \mu_B)\rangle,\qquad\langle\$_A(\nu_A, \mu_{B_i}) \xrightarrow[i]{} \langle\$_A(\nu_A, \mu_B)\rangle$$ because $\langle\$_A\rangle$ and $\langle\$_A(\nu_A, \cdot)\rangle$ are weakly$^*$-continuous on $\mathcal{M}$ (see Lemma \ref{contw}).\\
As a consequence we obtain
$$\langle\$_A(\mu_A, \mu_B)\rangle \geq \langle\$_A(\nu_A, \mu_B)\rangle \quad \forall\  \nu_A \in \mathcal{M},$$
giving $\langle\$_A(\mu_A, \mu_B)\rangle = \max_\mathcal{M} \langle\$_A(\cdot, \mu_B)\rangle$.
\end{proof}

\begin{remark}\rm
Theorem \ref{esistenza} generalises the known existence result for discrete mixed quantum strategies \cite{lee} to the continuous setting, where players may choose arbitrary probability measures on $\mathrm{SU}(2)$. It extends beyond the $2\times2$ case. Its proof uses only the weak*-compactness of the spaces of mixed strategies and the continuity of the payoff functions. Hence, it applies to any quantum game satisfying these conditions, including games with more than two strategies per player or different Hilbert space structures.
\end{remark}

Since ${\rm Fix}\mathcal{B}$ is precisely the set of mixed strategies that are Nash equilibria, we immediately obtain the following:
\begin{corollary} \label{cor:Nash}
    Any static $2\times 2$ quantum game has at least a Nash equilibrium in quantum mixed strategies.
\end{corollary}

\section{The Eisert-Wilkens-Lewenstein protocol}
The most widely used framework for implementing quantum versions of simultaneous $2\times2$ games is the protocol introduced by Eisert, Wilkens and Lewenstein (EWL).  See \cite{eisert_quantum} for further details.

It proceeds as follows:
\begin{enumerate}
\item Start from the state \(\ket{\varphi_0}=\ket{00}\).
\item Apply a unitary operator \(J\) that introduces a controlled amount of entanglement between the qubits, obtaining \(\ket{\varphi_1}=J\ket{\varphi_0}\).
\item Each player applies his own unitary strategy \(U_A,U_B\in\SU(2)\) to his qubit:
   \(\ket{\varphi_2}=(U_A\otimes U_B)\ket{\varphi_1}\).
\item Apply the disentangling operator \(J^\dagger\):
   \(\ket{\varphi_F}=J^\dagger\ket{\varphi_2}\).
\item Measure \(\ket{\varphi_F}\) in the computational basis \(\{\ket{00},\ket{01},\ket{10},\ket{11}\}\).
\end{enumerate}   

After the players' moves the state is then 
\begin{equation}\label{final-state}
\ket{\varphi_F}=J^\dagger(U_A\otimes U_B)J\ket{00},
\end{equation}
the probability of outcome \((i,j)\) is $p_{ij}=|\braket{ij|\varphi_F}|^2$, and the expected payoffs are
   \[
   \langle\$_{A}\rangle=\sum_{i,j\in\{0,1\}} a_{ij}p_{ij},\qquad
   \langle\$_{B}\rangle=\sum_{i,j\in\{0,1\}} b_{ij}p_{ij}.
\]

Note that this protocol yields the quantum game $\mathbf{G}$ with initial state $\rho=J|00\rangle\langle 00|J^\dagger$ and payoff operators 
\begin{equation}
    \label{PX-EWL}
P_X:=\sum_{i,j\in\{0,1\}}x_{ij}J|ij\rangle\langle ij|J^\dagger
\end{equation}
associated with the PVM $\{J\ket{ij}\,|\,i,j=0,1\}$. Indeed, by setting $\ket{\psi_{ij}}=J\ket{ij}$ and $\ket{v_f}=J\ket{\varphi_F}=(U_A\otimes U_B)J\ket{00}$, we have
\begin{align*}\langle\$_{X}\rangle&=\sum_{i,j} x_{ij}|\langle ij|\varphi_F\rangle|^2=\sum_{i,j} x_{ij}|\langle\psi_{ij}|(U_A\otimes U_B)J\ket{00}|^2\\
&=\sum_{i,j} x_{ij}|\langle\psi_{ij}|v_f\rangle|^2,
\end{align*}
which is precisely the expected payoff $\langle\$_{X}(U_A,U_B)\rangle$ of the game $\mathbf{G}$ by equation \eqref{eq:payoff-puro}, since the corresponding final state is the pure state
$$\rho_f(U_A, U_B)=(U_A\otimes U_B)J|00\rangle\langle 00|J^\dagger(U_A\otimes U_B)^\dagger=|v_f\rangle\langle v_f|.$$

   \subsection{Construction of the entangler $J$}
The entangling operator $J$ is designed to satisfy two fundamental requirements:
\begin{enumerate}
\item The game must be symmetric under exchange of the two players.
\item The classical game must be contained as a special case, i.e.~when the players are limited to classical moves the outcome probabilities should coincide with the classical ones.
\end{enumerate}
These requirements determine the form of \(J\) that prepares the initial quantum state. First of all, to ensure that the classical game is embedded, \(J\) must commute with any tensor product of classical moves (the identity \(\sigma_0:=I\) and the flip \(F=\mi\sigma_x\)). Since $\{\sigma_0, \mi\sigma_x, \mi\sigma_y, \mi\sigma_z\}$ is a basis of $\SU(2)$ and $[\sigma_j,\sigma_x]=0$ if and only if $j=0,x$, this condition gives
$$J=\sum_{i,j=0,x}J_{ij}\sigma_i\otimes\sigma_j.$$ 
Moreover, the symmetry of the game forces $J$ to commute with the swap operator that exchanges the two players, i.e. it must be of the form $J=J_{00}\,I\otimes I+J_{xx}\,\sigma_x\otimes\sigma_x.$ Finally, by requiring $J$ belonging to $\SU(2)$, we get
\begin{align}\label{defJ}
J&=\cos\frac{\gamma}{2}\,I\otimes I + \mi\sin\frac{\gamma}{2}\,\sigma_x\otimes\sigma_x=\exp\!\Bigl(\mi\frac{\gamma}{2}\,\sigma_x\otimes\sigma_x\Bigr)\\
\nonumber&=\begin{pmatrix}
\cos\frac{\gamma}{2} & 0 & 0 & \mi\sin\frac{\gamma}{2}\\[2pt]
0 & \cos\frac{\gamma}{2} & \mi\sin\frac{\gamma}{2} & 0\\[2pt]
0 & \mi\sin\frac{\gamma}{2} & \cos\frac{\gamma}{2} & 0\\[2pt]
\mi\sin\frac{\gamma}{2} & 0 & 0 & \cos\frac{\gamma}{2}
\end{pmatrix},
\qquad\gamma\in[0,\pi/2],
\end{align}
where the parameter \(\gamma\) controls the entanglement: \(\gamma=0\) gives no entanglement, \(\gamma=\pi/2\) maximal entanglement.

In particular, 
\begin{equation}\label{eq:psi_1}
\ket{\psi_1}=J\ket{00}= \cos\frac{\gamma}{2}\,\ket{00}+\mi\sin\frac{\gamma}{2}\,\ket{11}.
\end{equation}
\\

Note that The EWL protocol contains the classical game as a special case.  
\begin{proposition}[Classical subgame]\label{prop:classical}
The classical game is recovered within the EWL protocol in two natural ways:
\begin{itemize}
\item If $\gamma=0$ (no entanglement) the EWL protocol reproduces the classical game for any choice of $U_A,U_B\in\SU(2)$.
\item If the players are restricted to strategies of the form $U(\theta,0,0)$, the outcome probabilities factorise as in the classical mixed-strategy game, independently of $\gamma$.
      
\end{itemize}
\end{proposition}

\begin{proof}
Since $J$ commutes with the group $\{U(\theta,0,0)\,|\,\theta\in[0,\pi]\}$ generated by $I$ and $\sigma_x$, for $U_A=U(\theta_A,0,0)$ and $U_B=U(\theta_B,0,0)$ the probabilities $p_{ij}$ factorize as
\[|\braket{ij|\varphi_F}|^2=|\braket{i|U_A|0}|^2\,|\braket{j|U_B|0}|^2.\]  
Since $|\braket{0|U(\theta,0,0)|0}|^2=\cos^2(\theta/2)$ and $|\braket{1|U(\theta,0,0)|0}|^2=\sin^2(\theta/2)$, setting $p=\cos^2(\theta_A/2)$ and $q=\cos^2(\theta_B/2)$ yields exactly the classical mixed-strategy probabilities.  

The case $\gamma=0$ follows because $J=I$ makes the probabilities factorise for any $U_A,U_B\in \SU(2)$.
\end{proof}
Hence, the classical mixed strategies are exactly replicated by the one-parameter family \(U(\theta,0,0)\); the entanglement operator \(J\) commutes with such strategies and therefore does not generate any non-classical correlation.

\begin{remark}\rm
If player $A$ plays pure quantum strategies $U_A^{(1)},\ldots,U_A^{(N)}$ with probabilities $p_1,\ldots,p_N$, and player $B$ plays $U_B^{(1)},\ldots,U_B^{(L)}$ with probabilities $q_1,\ldots,q_L$, then the final state of the game $\mathbf{G}$ is a mixed state,
\[
\rho_f=\sum_{i=1}^N\sum_{j=1}^L p_i q_j (U_A^{(i)}\otimes U_B^{(j)})J\ket{00}\bra{00}J((U_A^{(i)})^\dagger\otimes (U_B^{(j)})^\dagger),
\]
and
\begin{equation}\label{eq:payoff-misto}
\langle\$_X\rangle=\Tr(\rho_f P_X)=\sum_{i=1}^N\sum_{j=1}^L p_i q_j \langle\$_X(U_A^{(i)},U_B^{(j)})\rangle
\end{equation}
for $X=A,B$, where $\langle\$_X(U_A^{(i)},U_B^{(j)})\rangle$ denotes the expected payoff of player $X$ when $A$ plays $U_A^{(i)}$ and $B$ adopts $U_B^{(j)}$, and $P_X$ is given by (\ref{PX-EWL}).\end{remark}

\subsection{Nash equilibria in the EWL protocol}
The general existence result for quantum mixed-
strategy Nash equilibria (Corollary \ref{cor:Nash}) applies directly to the EWL scheme.

\begin{theorem}[Existence of Nash equilibria in the EWL protocol]\label{esistenza-EWL}
For any static $2\times2$ game implemented through the EWL protocol with an arbitrary entanglement parameter $\gamma$, there exists at least one Nash equilibrium in quantum mixed strategies.
\end{theorem}

This result guarantees that, regardless of the entanglement level, the quantum game always possesses strategic equilibria. 
Their concrete form, however, depends crucially on $\gamma$ and can differ radically from the classical Nash equilibria.
\begin{remark}\rm
For maximally entangled EWL games ($\gamma=\pi/2$), Landsburg \cite{landsburg} developed a quaternionic formalism that reduces mixed strategies to those supported on at most four pure strategies. This reduction relies crucially on a symmetry property of the maximally entangled state that allows unitary operations to be transferred from one qubit to the other. For $\gamma\neq\pi/2$, this symmetry is broken and the quaternionic reduction fails. This justifies the need for a general existence theorem covering all entanglement levels.
\end{remark}

\subsection{Payoff as a quadratic form}

Within the EWL protocol, using the real-vector representation \(u_A,u_B\in\mathbb{R}^4\) of the strategies \(U_A,U_B\),the expected payoffs of players can be written as quadratic forms in $u_A$ and $u_B$.

\begin{theorem}[Payoff as a quadratic form]\label{thm:quadratic}
For any choice of \(U_B\) (equivalently \(u_B\)), there exists a self-adjoint \(4\times4\) matrix \(P_A(u_B)\) such that
\[
\langle\$_A(u_A,u_B)\rangle = u_A^{\mathsf T}\,P_A(u_B)\,u_A=\scalar{u_A}{P_A(u_B)u_A}.
\]
The matrix is given by
\[
P_A(u_B)=\sum_{j,k=0}^{1} a_{jk}\; |M_{jk}(u_B)\rangle\langle M_{jk}(u_B)|,
\]
where the vectors \(M_{jk}(u_B)\in\mathbb{C}^4\) are obtained from the expansion of \(\langle jk|\varphi_F\rangle\) as a linear function of \(u_A\).  
Explicitly,
\begin{align*}
M_{00}(u_B)&=\begin{pmatrix}u_B^1+\mi u_B^4\cos\gamma\\ -u_B^3\sin\gamma\\ -u_B^2\sin\gamma\\ -u_B^4+\mi u_B^1\cos\gamma\end{pmatrix},
&
M_{01}(u_B)&=\begin{pmatrix}\mi u_B^2-u_B^3\cos\gamma\\ \mi u_B^4\sin\gamma\\ \mi u_B^1\sin\gamma\\ -u_B^2\cos\gamma-\mi u_B^3\end{pmatrix},\\[6pt]
M_{10}(u_B)&=\begin{pmatrix}\mi u_B^3\sin\gamma\\ \mi u_B^1-u_B^4\cos\gamma\\ -\mi u_B^4-u_B^1\cos\gamma\\ \mi u_B^2\sin\gamma\end{pmatrix},
&
M_{11}(u_B)&=\begin{pmatrix}u_B^4\sin\gamma\\ -u_B^2-\mi u_B^3\cos\gamma\\ u_B^3-\mi u_B^2\cos\gamma\\ u_B^1\sin\gamma\end{pmatrix}.
\end{align*}

In particular, 
\begin{equation}\label{payoff-EWL}
\langle\$_A(u_A,u_B)\rangle=\sum_{l,k=0,1}a_{lk}|\scalar{u_A}{M_{lk}(u_B)}|^2.
\end{equation}

\end{theorem}
\begin{proof}
Write the state $\ket{\varphi_F}$ as
\[
\ket{\varphi_F}=J^{\dagger}\!\Bigl[U_A\ket{0}\otimes U_B\ket{0}\cos\frac\gamma2
+\mi\,U_A\ket{1}\otimes U_B\ket{1}\sin\frac\gamma2\Bigr].
\]
Expressing \(U_A\ket{0},U_A\ket{1}\) in terms of the components of \(u_A=(u_A^1,u_A^2,u_A^3,u_A^4)\),
\[
U_A\ket{0}=(u_A^1+\mi u_A^4)\ket{0}+(\mi u_A^2-u_A^3)\ket{1},\quad
U_A\ket{1}=(\mi u_A^2+u_A^3)\ket{0}+(u_A^1-\mi u_A^4)\ket{1},
\]
and similarly for \(U_B\), a direct computation shows that each amplitude  \(\langle jk|\varphi_F\rangle\) is a linear function of \(u_A\):
\[
\langle jk|\varphi_F\rangle= u_A^{\mathsf T}M_{jk}(u_B)=\langle u_A,M_{jk}(u_B)\rangle.
\]
Hence
\begin{align*}
\langle\$_A\rangle&=\sum_{j,k}a_{jk}|\langle jk|\varphi_F\rangle|^2
      =\sum_{j,k}a_{jk}\bigl|\langle u_A,M_{jk}(u_B)\rangle\bigr|^2\\
      &=\scalar{u_A}{\!\sum_{j,k}a_{jk}|M_{jk}(u_B)\rangle\langle M_{jk}(u_B)|u_A},
\end{align*}
which is precisely \(\scalar{u_A}{P_A(u_B)u_A}=u_A^{\mathsf T}P_A(u_B)u_A\).
\end{proof}

\begin{remark}[Symmetric form for player $B$]\rm\,
Analogously, there exists a self-adjoint matrix \(P_B(u_A)\) such that
\[
\langle\$_B(u_A,u_B)\rangle = \scalar{u_B}{P_B(u_A)\,u_B}.
\]
The vectors \(N_{jk}(u_A)\) defining \(P_B(u_A)\) are obtained from the same functional expressions as the \(M_{jk}\) with the roles of \(u_A\) and \(u_B\) exchanged; in fact
\begin{equation}\label{N00_M00}
N_{00}(u_A)=M_{00}(u_A),\quad N_{01}(u_A)=M_{10}(u_A),
\end{equation}
\begin{equation}\label{N11_M11}
N_{10}(u_A)=M_{01}(u_A),\quad N_{11}(u_A)=M_{11}(u_A).
\end{equation}
\end{remark}

\begin{remark}[Nash equilibrium condition]\rm\,
A pair \((u_A^*,u_B^*)\) is a Nash equilibrium in pure strategies if and only if
\(u_A^*\) is an eigenvector of \(P_A(u_B^*)\) corresponding to its largest eigenvalue,
and simultaneously \(u_B^*\) is an eigenvector of \(P_B(u_A^*)\) corresponding to its largest eigenvalue.
Those eigenvalues are exactly the payoffs \(\langle\$_A\rangle\) and \(\langle\$_B\rangle\) at the equilibrium.
\end{remark}

The quadratic-form representation provides a compact analytic tool for studying equilibria: the search for a Nash equilibrium reduces to finding a pair of unit vectors that are mutually optimal eigenvectors of the payoff matrices induced by each other's strategy.

\subsection{Example: The Chicken game}
In this subsection, we demonstrate the usefulness of Theorem \ref{thm:quadratic} in calculating the players' payoffs when the scenario is asymmetric, i.e. player $A$ chooses classical strategies while player $B$ has access to the full $\SU(2)$. In particular, we prove that player $B$ can always improve their payoff by adopting quantum pure strategies.

 In this setting, let $u_A=(\cos\frac{\varphi}{2},\sin\frac{\varphi}{2},0,0)$ and 
\[
u_B =(\cos\alpha\cos\frac\theta2,\cos\beta\sin\frac\theta2, -\sin\beta\sin\frac\theta2,\sin\alpha\cos\frac\theta2)=(w,x,y,z)
\]
be the unit vectors in $\mathbb{R}^4$ corresponding to strategies $U_A=U(\varphi,0,0)$ and $U_B=U(\theta, \alpha, \beta)$, respectively. Using the scalar product expression for the payoff given by the Theorem \ref{thm:quadratic}, we get
\begin{align}
 \nonumber \langle\$_X(u_A,u_B)\rangle&= x_{00}\left[z^2\cos^2\frac{\varphi}{2}\cos^2\gamma+(w\cos\frac{\varphi}{2}-y\sin\frac{\varphi}{2}\sin\gamma)^2
  \right]\\
  \nonumber &+x_{01}\left[y^2\cos^2\frac{\varphi}{2}\cos^2\gamma+(z\sin\frac{\varphi}{2}\sin\gamma+x\cos\frac{\varphi}{2})^2
  \right]\\
  \nonumber &+x_{10}\left[z^2\sin^2\frac{\varphi}{2}\cos^2\gamma+(y\cos\frac{\varphi}{2}\sin\gamma+w\sin\frac{\varphi}{2})^2
  \right]\\
 \label{formula-payoff-asimm} &+x_{11}\left[y^2\sin^2\frac{\varphi}{2}\cos^2\gamma+(z\cos\frac{\varphi}{2}\sin\gamma-x\sin\frac{\varphi}{2})^2
  \right]
\end{align}
for $X=A,B$.

\begin{example}\rm
The Chicken game (also known as Hawk-Dove) models a conflict where each player can choose an aggressive (Hawk) or conciliatory (Dove) behaviour. Mutual aggression leads to a severe loss for both, while a unilateral aggressive move yields the highest individual payoff. Mutual conciliation gives a moderate, positive outcome. Classically, the game has two pure asymmetric Nash equilibria (one Hawk, one Dove) and one mixed equilibrium. The payoff matrix used here is standard in the literature.

Consider now the Chicken game with the following payoff matrix (player $A$ rows, player $B$ columns):
\[
\begin{array}{c|cc}
 & H & D \\
\hline
H & (-25,-25) & (50,0) \\
D & (0,50) & (15,15)
\end{array}
\]
Thus, for player $A$: $a_{00}=-25$, $a_{01}=50$, $a_{10}=0$, $a_{11}=15$;
for player $B$: $b_{00}=-25$, $b_{01}=0$, $b_{10}=50$, $b_{11}=15$.

Substituting these values into the general formula \eqref{formula-payoff-asimm} and simplifying, we obtain the expected payoffs for the two players. In particular, the difference between them is
\[
\begin{aligned}
\langle \$ _A \rangle - \langle \$ _B \rangle =&\; 50\big[\cos^2\!\gamma\big(y^2\cos^2\frac{\varphi}{2} - z^2\sin^2\frac{\varphi}{2}\big) \\
&+ \bigl(x\cos\tfrac{\varphi}{2} + z\sin\tfrac{\varphi}{2}\sin\gamma\bigr)^2 \\
&- \bigl(w\sin\tfrac{\varphi}{2}+y \cos\tfrac{\varphi}{2}\sin\gamma\bigr)^2].
\end{aligned}
\]

This expression shows that for suitable choices of the quantum strategy parameters $(w,x,y,z)$ (i.e. of $u_B$), the difference $\langle \$ _A \rangle - \langle \$ _B \rangle$ can be made non-positive. We now demonstrate this explicitly.

\begin{proposition}
For all $\gamma\in[0,\pi/2]$ and for every classical strategy $U_A=U(\varphi,0,0)$ chosen by $A$, there exists at least one quantum pure-strategy $U_B\in\SU(2)$ such that
\[
\langle \$ _A(U_A,U_B)\rangle \le \langle \$ _B(U_A,U_B)\rangle,
\]
with strict inequality except at $\varphi=0$ and $\gamma\leq\pi/4$.
\end{proposition}
\begin{proof}

\textbf{Case 1: $\mathbf{\gamma\neq\pi/2}$ and $\mathbf{\varphi\neq 0}$.}
Choose $w=\sin\frac{\varphi}{2}\cos\gamma,\ x=\sin\frac{\varphi}{2}\sin\gamma,\ y=0$ and $z=-\cos\frac{\varphi}{2}$. Substituting into the difference
\begin{equation*}
\langle \$ _A \rangle - \langle \$ _B \rangle = -50\cos^2\gamma\sin^2\frac{\varphi}{2}\left(\cos^2\frac{\varphi}{2}+\sin^2\frac{\varphi}{2}\right) = -50\cos^2\gamma\sin^2\frac{\varphi}{2},
\end{equation*}
which is strictly smaller than zero if and only if $\gamma\neq\pi/2$ and $\varphi\neq 0$ .
Hence the quantum player $B$ obtains a strictly higher payoff.

{\textbf{Case 2:}} $\mathbf{\gamma = \pi/2}$.
We have 
\[
\langle \$ _A \rangle - \langle \$ _B \rangle = 50\left( \bigl(x\cos\tfrac{\varphi}{2} + z\sin\tfrac{\varphi}{2}\bigr)^2
- \bigl(w\sin\tfrac{\varphi}{2}+y \cos\tfrac{\varphi}{2}\bigr)^2\right),
\] so we can choose $x=0=z$ and $w=\frac{1}{\sqrt{2}}=y$ in order to have the difference strictly negative for all $\varphi\in[0,\pi]$.

{\textbf{Case 3: $\varphi = 0$ (player $A$ plays pure $H$)}}. The payoff difference reduces to
\[
\langle \$ _A \rangle - \langle \$ _B \rangle = 50\left( (1 - 2\sin^2\gamma)\, y^2 + x^2 \right).
\]
The coefficient $1 - 2\sin^2\gamma$ is positive for $\gamma < \pi/4$, zero for $\gamma = \pi/4$, and negative for $\gamma > \pi/4$.
Therefore, we have two subcases:
\begin{itemize}
  \item If $\gamma > \pi/4$, choose $u_B=(0,0,1,0)$, so that 
  $$\langle \$ _A \rangle - \langle \$ _B \rangle=50(1 - 2\sin^2\gamma)<0.
  $$
  \item If $\gamma \leq \pi/4$, it is not guaranteed that the difference is strictly negative, but we have $\langle \$ _A \rangle = \langle \$ _B \rangle$ by choosing $x = y = 0$. Hence, for $\gamma \le \pi/4$ and $\varphi = 0$, no choice of the quantum strategy yields a strictly negative difference.
\end{itemize}

Thus, in all cases we have constructed explicit quantum strategies for $B$ ensuring $\langle \$ _B \rangle \ge \langle \$ _A \rangle$, with strict inequality except possibly at $\varphi=0$. This completes the proof.
\end{proof}
\end{example}

\medskip

%{\acknowledgements The authors are members of GNAMPA-INdAM. The authors have been supported by the MUR grant \lq\lq Dipartimento di Eccellenza 2023--2027\rq\rq\, of Dipartimento di Matematica, Universit\`a di Genova.}

\vskip 1truecm

Authors' addresses:
\begin{itemize}
\item[$^{(1)}$] Department of Mechanical, Energy, Management and Transportation Engineering, University of Genoa, Via all'Opera Pia, 15 16145 Genova, Italy, email: gloria.ferraris@edu.unige.it
\item[$^{(2)}$] Mathematics Department, University of Genova,
Via Dodecaneso, I - 16146 Genova, Italy, email: veronica.umanita@unige.it
\end{itemize}


\begin{thebibliography}{9}

% Questo perche' appaia nell'indice:
\addcontentsline{toc}{chapter}{Bibliografia}

\bibitem{arandelovic2000}
{\textsc I. D. Arandelovi\'c},
A new extension of Kakutani's fixed point theorem,
{\it Annals of University of Timi\c{c}oara} {\bf 38}, No. 1 (2000).

\bibitem{benjamin}
{\textsc S. C. Benjamin}, Comment on: A quantum approach to static games of complete information, https://xxx.lanl.gov/abs/quant-
ph/0008127

\bibitem{Hayden}
{\textsc  S. C. Benjamin and P. M. Hayden}, 
Multiplayer quantum games, {\it Phys. Rev. A} {\bf 64}, 030301 (2001).

\bibitem{bleiler}
{\textsc S. A. Bleiler},
A Formalism for Quantum Games and an Application, \\
https://doi.org/10.48550/arXiv.0808.1389, (2008).

\bibitem{bolonek1}
{\textsc K. Bolonek-Laso\'{n}}, General quantum two-player games, their gate operators, and Nash equilibria, {\it Progress of Theoretical and Experimental Physics} {\bf 2015.2}, 023A03 (2015).

\bibitem{bolonek2}
{\textsc K. Bolonek-Laso\'{n} and P. Kosi\'{n}ski}, Mixed Nash equilibria in Eisert-Lewenstein-Wilkens (ELW) games, {\it J. Phys.: Conf. Ser.} {\bf 804}, 012007 (2017).

\bibitem{chappell}
{\textsc J. M. Chappell, A. Iqbal and D. Abbott}, $N$-player quantum games in an EPR setting, {\it  PLoS One} {\bf 7}(5), e36404 (2012).

\bibitem{das2023}
{\textsc S. Das},
Quantumizing Classical Games: An Introduction to Quantum Game Theory (2023).
https://doi.org/10.48550/arXiv.2305.00368


\bibitem{eisert_2009}
{\textsc J. Eisert and M. Wilkens},
Quantum games,
{\it Journal of Modern Optics} {\bf 47}, (14-15)  2543-2556 (2000).  https://doi.org/10.1080/09500340008232180.

\bibitem{eisert_quantum}
{\textsc J. Eisert, M. Wilkens, and M. Lewenstein},
Quantum games and quantum strategies,
{\it Phys. Rev. Lett.} {\bf 83}, 3077 (1999).

\bibitem{lee}
{\textsc C. F. Lee and N. F. Johnson},
Efficiency and formalism of quantum games, {\it Phys. Rev. A} {\bf 67}, 022311 (2003). 

\bibitem{flitney_abbott}
{\textsc A. P. Flitney and D. Abbott},
An introduction to quantum game theory,
{\it Fluctuation and Noise Letters} {\bf 2}, No. 4, R175--R187 (2002).

\bibitem{flitney_hollenberg}
{\textsc A. P. Flitney and L.C.L. Hollenberg},
Nash equilibria in quantum games with generalized two-parameter strategies,
{\it Physics Letters A} {\bf 363}, No. 5-6, 381-388 (2007).

\bibitem{landsburg}
{\textsc S.E. Landsburg}
Nash equilibria in quantum games, {\it Proc. Am. Math. Soc.} {\bf 139}, No. 12 , 4423-4434 (2011).

\bibitem{Jiangfeng_Du_2003}
{\textsc D. Jiangfeng, L. Hui, X. Xiaodong, Z. Xianyi and H. Rongdian}, Phase-transition-like behaviour of quantum games, {\it Journal of Physics A: Mathematical and General} {\bf 36}(23), 6551 (2003).

\bibitem{marinatto_weber}
{\textsc L. Marinatto and T. Weber}, A quantum approach to static games of complete information, {\it Physics Letters A} {\bf 272}, 291-303 (2000).

\bibitem{meyer_penny}
{\textsc D. A. Meyer},
Quantum strategies,
{\it Phys. Rev. Lett.} {\bf 82}, 1052 (1999).

\bibitem{nash_equilibrium}
{\textsc J. F. Nash},
Equilibrium points in $n$-person games,
{\it Proceedings of the National Academy of Sciences} {\bf 36}(1), 48-49 (1950).

\bibitem{von_neumann_morgenstern}
{\textsc J. von Neumann and O. Morgenstern},
{\it Theory of Games and Economic Behavior},
Princeton University Press (1944).
\end{thebibliography}
\end{document}